\documentclass[11pt]{article}

\usepackage[letterpaper,margin=1in]{geometry}
\usepackage[T1]{fontenc}
\usepackage{newpxtext}
\usepackage{mathtools,amssymb,amsthm}
\usepackage[vvarbb]{newtxmath}
\usepackage{microtype}
\usepackage{xcolor}
\usepackage{url}
\usepackage[bookmarks=true,hypertexnames=false]{hyperref}
\usepackage{thmtools}
\usepackage[nameinlink,capitalise,noabbrev]{cleveref}
\usepackage{braket}
\newtheorem{theorem}{Theorem}[section]
\newtheorem{lemma}[theorem]{Lemma}

\theoremstyle{definition}

\theoremstyle{remark}

\newcommand{\bits}{\{0,1\}}
\newcommand{\cube}{\{0,1\}}
\newcommand{\R}{\mathbb{R}}
\newcommand{\A}{\mathsf{A}}
\newcommand{\B}{\mathsf{B}}
\newcommand{\KA}{K_{\A}}
\newcommand{\KB}{K_{\B}}
\newcommand{\KE}{K_{\mathsf E}}
\newcommand{\cK}{\mathcal{K}}
\newcommand{\cL}{\mathcal{L}}
\newcommand{\cR}{\mathcal{R}}
\DeclareMathOperator{\supp}{supp}
\DeclareMathOperator{\ml}{ml}

\title{Impossibility of Perfectly Complete \\ Many-Round Key Agreement in the QROM}

\author{
Longcheng Li\thanks{University of Cambridge. Email: \texttt{lilongcheng116@gmail.com}}
\quad
Qian Li\thanks{Shenzhen Research Institute of Big Data. Email: \texttt{liqian.ict@gmail.com}}
\quad
Xingjian Li\thanks{Tsinghua University. Email: \texttt{lxj22@mails.tsinghua.edu.cn}}
\quad
Qipeng Liu\thanks{University of California San Diego. Email: \texttt{qipengliu0@gmail.com}}
}

\date{}

\begin{document}

\maketitle

\begin{abstract}
This paper proves that it is impossible to construct perfectly complete quantum key agreement protocols (QKA) from quantumly secure one-way functions (OWFs) in a black-box manner. 

Specifically, consider any protocol in which Alice and Bob exchange only classical messages, make at most $q_{\A}$ and $q_{\B}$ quantum queries, respectively, to a Boolean-valued random oracle, and agree on a shared key with certainty. This paper shows that there exists an eavesdropper, given the classical messages, that can recover the shared key with certainty using $O((q_{\A}+q_{\B})^5)$ classical oracle queries. The bound is independent of the number of rounds, transcript length, key length, and oracle-domain size. Previous results only applies to two-round key agreement (Li et al. CRYPTO 26) or relies on unproven conjectures (Austrin et al. CRYPTO 22). 

GPT-5.6 Sol Ultra found this proof in a one-shot conversation and drafted a preliminary version of this paper. The authors are fully responsible for the correctness, writing and discussions of this paper.
\end{abstract}

\section{Introduction}

Key agreement is a fundamental cryptographic task in which two parties use public interaction to establish a shared secret that was not available to either party in advance. In the random oracle model (ROM), Impagliazzo--Rudich established the foundational black-box separation between one-way functions (OWFs) and key agreement, and Barak--Mahmoody later gave an \(O(q^2)\)-query eavesdropping attack against every \(q\)-query classical key-agreement protocol~\cite{ImpagliazzoRudich89,BarakMahmoody09}.
The quantum random oracle
model (QROM) is more subtle because an honest party may query a superposition
of oracle inputs~\cite{BonehEtAl11}.  A particularly interesting case is
\emph{quantum computation with classical communication} (QCCC): the local
computations and oracle queries are quantum, while every communicated message
is classical.

Austrin, Chung, Chung, Fu, Lin, and Mahmoody made the first significant progress on QCCC
key agreement in the QROM\@.  They gave an unconditional quadratic-query
classical attack when one party is classical and an unconditional
exponential-query classical attack when both parties are quantum.  Under their Polynomial Compatibility Conjecture,
they obtained a polynomial-query classical attack against general perfectly complete QCCC protocols~\cite{AustrinEtAl22}.

A subsequent sequence of works by Li, Li, Li, and Liu studied the two-message
key-agreement protocols induced by quantum public-key encryption.  Their first
work uses quantum Markov chains and handles classical-query key
generation~\cite{LiEtAl24}.  Their second work formulates a Boolean-function
conjecture and unconditionally handles logarithmically many quantum queries
in key generation~\cite{LiEtAl25}.  Most recently, the conjecture and the
restriction on key-generation queries were both removed for perfectly
complete quantum public-key encryption with classical keys and classical or
quantum ciphertexts, using a quantum-query
eavesdropper~\cite{LiEtAl26}.  These developments leave the natural question
of whether arbitrary-round, perfectly complete QCCC key agreement can be
secure in the QROM. 

\subsection{Our Main Result}

Our main theorem gives a round-independent negative answer, improving the two-round result in~\cite{LiEtAl26} and removing the conjecture in~\cite{AustrinEtAl22}.  More interestingly,
although the honest parties may query the oracle in superposition, the
eavesdropper only makes classical queries.

\begin{theorem}[Informal]\label{thm:informal}
Fix a Boolean-output random oracle.  Every finite-round, perfect complete
QCCC key-agreement protocol admits a query-efficient eavesdropper that recovers
the shared key with probability one.  If the honest parties make \(q_{\A}\)
and \(q_{\B}\) quantum oracle queries, the eavesdropper makes
\(O((q_{\A}+q_{\B})^5)\) classical oracle queries.
\end{theorem}

The query bound is independent of the number of rounds, the transcript
length, the key length, and the oracle-domain size. Relative to the
polynomial-query attack of~\cite{AustrinEtAl22}, our attack is unconditional and
recovers the key with certainty. Relative to the later QPKE
results~\cite{LiEtAl26}, our attack applies to arbitrarily many classical communication rounds and uses only classical queries. 
Moreover, our new result allows parties to keep arbitrary quantum states between stages; \cite{LiEtAl26} handles only the two-round case where parties keep classical states between stages (corresponding to classical secret keys in a quantum public-key encryption scheme).

The result is information-theoretic in query complexity.  The eavesdropper is
allowed unbounded computation, but only bounded number of quantum/classical queries.
We assume that two parties Alice and Bob begin with independent private states and
exchange only classical messages; pre-shared entanglement, correlated setup,
and quantum communication are not allowed in our model.  Perfect completeness is
essential to our argument.  Extending the proof to imperfect completeness
would require significantly different ideas.

\paragraph{Proof overview.}
Fix the complete public transcript \(t\), a pair of output keys \(a,b\), and
the truth table \(z\in\{0,1\}^N\) encoding the oracle.  Because communication
is classical and the initial private state is a product, the probability of
this branch factors as
\[
  \Pr[T=t,\KA=a,\KB=b\mid z]
  = A_{t,a}(z)\cdot B_{t,b}(z).
\]
The quantum polynomial method~\cite{BBCMW01} bounds the degrees of the two
nonnegative factors by \(2q_{\A}\) and \(2q_{\B}\).

Perfect completeness holds pointwise on the finite oracle space.  Consequently,
all cross-key products \(A_{t,k}B_{t,\ell}\), for \(k\ne\ell\), vanish on the
Boolean cube.  
Consider the following polynomials on $\{0,1\}^N$:
\[
  C_{t,k}=A_{t,k}\cdot B_{t,k}. 
%  \qquad\text{where}\qquad   d = 2(q_{\A}+q_{\B}).
\]
where have degree at most \(d:=2(q_\A+q_\B)\). For two different keys \(k\ne\ell\), perfect completeness gives
\[
\begin{aligned}
  C_{t,k}(z)C_{t,\ell}(z)
  &=A_{t,k}(z)B_{t,k}(z)
    A_{t,\ell}(z)B_{t,\ell}(z)\\
  &=\bigl(A_{t,k}(z)B_{t,\ell}(z)\bigr)
    \bigl(A_{t,\ell}(z)B_{t,k}(z)\bigr)
  =0.
\end{aligned}
\]
Thus, fixing transcript $t$, the supports of $\{C_{t, k}\}_{k}$ are pairwise disjoint.
Moreover, summing over the output keys gives
\[
  \Pr[T=t\mid z]=\sum_k C_{t,k}(z).
\]
Whenever \(t\) can occur under \(z\), the left-hand side is positive, so
exactly one term on the right is nonzero for each $z$.  Conditioned on \((t,z)\), both
parties therefore output the uniquely determined key indexing that term.

Finally, by the standard fact that a nonzero multilinear polynomial of degree at most \(d\) on \(\{0,1\}^N\) is nonzero on at least \(2^{N-d}\) points, we have that, %One way to see this is to choose a highest-degree monomial: after any fixing of the variables outside that monomial, the restricted polynomial remains nonzero and hence is nonzero at some point of that subcube. Thus, 
for the set
\(\cK_t=\{k:C_{t,k}\not\equiv0\}\) of keys that can occur with \(t\) under
some oracle,
\[
  2^N
  \ge \sum_{k\in\cK_t}|\supp(C_{t,k})|
  \ge |\cK_t|\,2^{N-d},
\]
since $\supp(C_{t,k})$ are disjoint. 
It follows that \(|\cK_t|\le2^d\).

It remains to identify the compatible key without evaluating the
\(C_{t,k}\)'s one by one.  We prove the following useful lemma.  Suppose that
two polynomials \(p\) and \(q\), each of degree at most \(d\), have disjoint supports, and we are promised that one of them is nonzero at
the actual oracle (i.e., the input $z$).  Then a deterministic procedure can determine which one by
reading only \(O(d^4)\) oracle bits.

The procedure proceeds in rounds. In each round, it identifies \(O(d^3)\)
oracle positions with the following property: once the values at these
positions are known, either the answer is determined immediately or
the degree of one of the two relevant polynomials decreases. Since the degree
can decrease at most \(O(d)\) times, the procedure determines which of the two
polynomials is nonzero at \(z\) using a total of \(O(d^4)\) classical oracle
queries.

Finally, the eavesdropper performs a balanced binary search over the at most
\(2^d\) candidate keys. At each level of the search, the candidates are
partitioned into two sets. Let \(p\) be the sum of the polynomials associated
with the first set, and let \(q\) be the corresponding sum for the second set.
Because exactly one candidate polynomial evaluates to a nonzero value at
\(z\), exactly one of \(p(z)\) and \(q(z)\) is nonzero. The procedure above
therefore determines which half contains the correct key using \(O(d^4)\)
classical oracle queries. Since the binary search has at most \(d\) levels,
the eavesdropper uses \(O(d^5)\) classical oracle queries in total.
The detailed proof combines ideas from~\cite{NS94}, \cite{midrijanis2004exact}, and~\cite[Theorem~4]{KothariEtAl26}. %and Li et al.~\cite[Lemmas~2.12 and~3.4]{LiEtAl26}.

\paragraph{Acknowledgment.} 
GPT-5.6 Sol Ultra discovered the proof in a one-shot conversation and produced a preliminary draft of the paper. The authors independently verified every statement and proof, simplified and refined the argument, and wrote the final manuscript. The authors take full responsibility for the correctness, exposition, and discussion presented in the paper.

\section{Preliminaries}\label{sec:preliminaries}

\subsection{Notations in Boolean Functions Analysis}
For a positive integer $N$, let $[N]=\{1,\ldots,N\}$.  Any function $f:\cube^N\rightarrow\mathbb{R}$ has a unique expression as a
multilinear polynomial
\[
f(x)=\sum_{S\subseteq[N]}a_S\cdot x_S,
\]
where $x_S:=\Pi_{i\in S} x_i$, and $a_S$ is the coefficient of $x_S$. 
The \emph{degree} of $f$, denoted $\deg(f)$, is defined as $\max\{|S|: a_S\neq 0\}$. A monomial $x_S$ is called \emph{maximum} if $a_S\neq 0$ and it has degree $\deg(f)$, i.e., $|S|=\deg(f)$. Two monomials $x_S$ and $x_T$ are called \emph{disjoint} if $S\cap T=\emptyset$. %We say that $f$ is \emph{not identically zero} if $f(x)\not\equiv0$.
We let $\supp(f)=\{x\in\cube^N:f(x)\neq0\}$ denote the support of $f$. For any polynomial $g$ with domain $\cube^N$, we have that $x_i^k=x_i$ for any $k\geq 1$. The multilinearization $\ml(g)$ of polynomial $g$ is obtained by replacing every $x_i^k$ with $x_i$ for $k\geq 1$.

A restriction $\rho$ fixes some variables to $\cube$.  We write
$f|_{\rho}$ for the function on the remaining variables obtained after the restriction.  For $x\in\cube^N$ and $S\subseteq[N]$, let $x^S$ be the point obtained by flipping the coordinates in $S$: $x^S_i=1-x_i$ for $i\in S$ and $x^S_i=x_i$ otherwise. For $x\in\cube^{N}$, we use $|x|:=|\{i:x_i=1\}|$ to denote the Hamming weight.

We will use the following standard facts in Boolean function analysis.

\begin{lemma}[\cite{NS94}]\label{lem:support}
Let $f:\cube^N\to\R$ be a non-zero function of degree at most $d$, then
$|\supp(f)|\ge 2^{N-d}$.
\end{lemma}

\begin{lemma}[\cite{BBCMW01}]\label{prop:polynomial-method}
    Suppose a quantum algorithm makes $d$ queries to a Boolean string\footnote{We interpret the string $x \in \{0,1\}^N$ as a function $x:[N] \to\{0,1\}$, and model queries to $x$ as oracle queries to this function.} $x\in\{0,1\}^N$,  and the acceptance probability is denoted by $f(x)$. Then the function $f:\{0,1\}^N\rightarrow \mathbb{R}$ has degree at most $2d$. The same holds for the squared norm of any fixed postselected branch. 
\end{lemma}

A deterministic decision tree adaptively queries coordinates of an unknown
$x\in\cube^N$.  Its depth is the maximum number of queried coordinates on
any root-to-leaf path.  

\subsection{The QROM and QCCC Key Agreement}

For notational simplicity, we use the security parameter \(\lambda\) also as
the oracle input length and write \(N=2^\lambda\).  A uniformly random oracle
is a uniformly random function
\[
  H:\bits^\lambda\longrightarrow\bits.
\]
We identify \(H\) with its truth table
\[
  \mathrm{tt}(H)=\bigl(H(x)\bigr)_{x\in\bits^\lambda}\in\cube^N.
\]
For \(z\in\cube^N\), let \(H_z\) denote the unique oracle having \(\mathrm{tt}(H_z)=z\).
We use the standard oracle query convention
\[
  O_{H_z}\lvert x,u\rangle
  =\lvert x,u\oplus H_z(x)\rangle,
  \qquad u\in\bits.
\]
In particular, a classical algorithm can recover
\(H_z(x)\) with one query by setting $u=0$.

A QCCC key-agreement protocol consists of two parties, Alice and Bob, that
begin with independent private states.  They may use private randomness,
arbitrary quantum computation, and arbitrary measurements, and they may query
\(O_{H_z}\) in superposition.  Every communicated message is classical.  Let
\(T\) be the complete public transcript, including any public randomness, and
let \(\KA,\KB\) be the classical output keys.  For each fixed \(\lambda\), we
assume that the protocol has finite transcript and key alphabets.  The number
of rounds has a finite worst-case bound for each \(\lambda\), but that bound
may depend arbitrarily on \(\lambda\).  All initial states and non-oracle
operations are independent of \(H\).  Alice and Bob make at most \(q_{\A}\)
and \(q_{\B}\) oracle queries. 
Without loss of generality, any external public coin is sampled by Alice and
announced as the first classical message.

All unconditioned probabilities are over the uniform oracle and all protocol
randomness and measurement outcomes.  Conditioned expressions such as
\(\Pr[\mathcal{E}\mid z]\) mean
\(\Pr[\mathcal{E}\mid \mathrm{tt}(H)=z]\).  The protocol is \emph{perfect complete} if
\[
  \Pr[\KA=\KB]=1.
\]

A passive eavesdropper observes \(T\) without altering the execution and may
then adaptively query the same oracle.  Our eavesdropper is required to be
query-efficient, but not computationally efficient or uniform.  It
may depend on the public protocol and \(\lambda\), but not on the sampled
oracle or private execution.  We count only oracle queries: there is no bound on running time or space, and all hidden constants in our query bounds are universal. 

\subsection{A Discrete Markov Inequality}

We will also use the following lemma due to the discrete Markov inequality from \cite{NS94}.
\begin{lemma}[\cite{KothariEtAl26}]\label{lem:discrete-markov}
Let $r:\{0,1\}^b\rightarrow\R$ be a polynomial with the following properties:
\begin{itemize}
\item $|r(x)|\leq 1$ for any $x\in\{0,1\}^b$,
\item $|r(0^b)|=1$,
\item $r(x)\cdot r(0^b)\leq 0$ for any $x\in\{0,1\}^b$ with $|x|=1$.
\end{itemize}
Then $\deg(r)\geq \sqrt{b/2}$.
\end{lemma}

\section{Proof of the Main Theorem}\label{sec:main}

We now prove the quantitative, pointwise version of
\Cref{thm:informal}.

\begin{theorem}[Main theorem]\label{thm:main}
Let \(\Pi\) be a perfectly complete QCCC key-agreement protocol in the QROM\@.
Suppose Alice and Bob begin with independent private states and make at most
\(q_{\A}\) and \(q_{\B}\) quantum queries, respectively.  Set
\[
  d=\min\{2(q_{\A}+q_{\B}),N\}.
\]
There is a passive eavesdropper against \(\Pi\)
with the following property: for every \(z\in\cube^N\) and every transcript
\(t\) satisfying \(\Pr[T=t\mid z]>0\), after observing \(t\), the eavesdropper
makes \(O(d^5)\) classical queries to \(H_z\) and outputs the unique key \(k\)
for which
\[
  \Pr[\KA=\KB=k\mid T=t,z]=1.
\]
Consequently,
\[
  \Pr[\KA=\KB=\KE]=1,
\]
and the eavesdropper uses
\(O((q_{\A}+q_{\B})^5)\) classical oracle queries.
\end{theorem}

\subsection{Low-Degree Transcript Rectangularity}

Fixing a complete classical transcript separates the two local quantum
executions.  This is the only step that uses the restriction to classical
communication.

\begin{lemma}[Transcript rectangularity]\label{lem:rectangularity}
For every complete transcript \(t\), pair of keys \(a,b\), and
\(z\in\cube^N\), there are functions
\(A_{t,a},B_{t,b}:\cube^N\to\R_{\ge0}\) such that
\begin{equation}
  \Pr[T=t,\KA=a,\KB=b\mid z]
  =A_{t,a}(z)B_{t,b}(z).
  \label{eq:rectangle}
\end{equation}
Moreover, their unique multilinear representations satisfy
\[
  \deg(A_{t,a})\le2q_{\A},
  \qquad
  \deg(B_{t,b})\le2q_{\B}.
\]
\end{lemma}

\begin{proof}

Purify the initial private states and private randomness, replace every local operation by a Stinespring isometry, and retain unrecorded Kraus indices in local environments.  
For fixed \((t,a)\), we compose Alice's
message-conditioned isometries and postselect on the outgoing-message projections, together
with her final-key projection. The operators composes  into a linear branch operator
\(M^{\A}_{t,a}(z)\).  Incoming messages act only as classical controls and are
hardwired to their values in \(t\).  Define \(M^{\B}_{t,b}(z)\) analogously.
Outgoing-message projections occur only at the sending party.

Let \(\lvert\psi_{\A}\rangle\otimes\lvert\psi_{\B}\rangle\) be the purified
initial product state and set
\[
  \lvert\alpha^z_{t,a}\rangle
  =M^{\A}_{t,a}(z)\lvert\psi_{\A}\rangle,
  \qquad
  \lvert\beta^z_{t,b}\rangle
  =M^{\B}_{t,b}(z)\lvert\psi_{\B}\rangle.
\]
Conditioning every classical channel on the value recorded in \(t\) leaves
the subnormalized global branch
\[
  \lvert\alpha^z_{t,a}\rangle\otimes
  \lvert\beta^z_{t,b}\rangle.
\]
Its squared norm is the branch probability and factors.  Thus
\cref{eq:rectangle} holds with
\[
  A_{t,a}(z)=\|\alpha^z_{t,a}\|^2,
  \qquad
  B_{t,b}(z)=\|\beta^z_{t,b}\|^2.
\]
These functions are nonnegative on the Boolean cube.  Applying
\Cref{prop:polynomial-method} to the two fixed local branches gives the stated
degree bounds.
\end{proof}

\subsection{Perfect Completeness and Disjoint Key Supports}

For every transcript \(t\) and key \(k\), define the function $C_{t,k}:\{0,1\}^N\rightarrow \R$ as
\begin{equation*}
  C_{t,k}=A_{t,k}\cdot B_{t,k}.
 % \label{eq:key-poly}
\end{equation*}
On the Boolean cube, \(C_{t,k}(z)=A_{t,k}(z)B_{t,k}(z)\ge0\), and
\(\deg(C_{t,k})\le d\).

\begin{lemma}[Disjoint key supports]\label{lem:disjoint-keys}
For every transcript \(t\), distinct keys \(k\ne\ell\), and
\(z\in\cube^N\),
\[
  C_{t,k}(z)C_{t,\ell}(z)=0.
\]
Moreover, whenever \(\Pr[T=t\mid z]>0\), exactly one \(C_{t,k}(z)\) is
nonzero, and both honest parties output that key with conditional probability
one.
\end{lemma}

\begin{proof}
Perfect completeness, the uniform distribution on the finite oracle space,
and \Cref{lem:rectangularity} give
\begin{equation*}
\begin{split}
0
  =\Pr[\KA\ne\KB]
  =2^{-N}\sum_{z\in\cube^N}
    \sum_t\sum_{a\ne b}A_{t,a}(z)B_{t,b}(z).
\end{split}
%\label{eq:pointwise}
\end{equation*}
Every summand is nonnegative, and all alphabets are finite.  Hence, 
for every \(z,t\) and \(a\ne b\),
\begin{equation}
  A_{t,a}(z)B_{t,b}(z)=0.
 \label{eq:cross-zero}
\end{equation}
For \(k\ne\ell\), it follows that
\[
\begin{aligned}
C_{t,k}(z)C_{t,\ell}(z)
  &=A_{t,k}(z)B_{t,k}(z)
    A_{t,\ell}(z)B_{t,\ell}(z)\\
  &=\bigl(A_{t,k}(z)B_{t,\ell}(z)\bigr)
    \bigl(A_{t,\ell}(z)B_{t,k}(z)\bigr)=0.
\end{aligned}
\]

Let \(A_t=\sum_kA_{t,k}\) and \(B_t=\sum_kB_{t,k}\).  Summing
\cref{eq:rectangle} over the two keys $a,b$ and using
\cref{eq:cross-zero}, we obtain
\begin{equation}
  \Pr[T=t\mid z]
  =A_t(z)B_t(z)
  =\sum_k C_{t,k}(z).
  \label{eq:transcript-sum}
\end{equation}
If the left-hand side is positive, at least one summand is positive.
Pairwise disjointness makes this key \(k\) unique.  Moreover,
\[
  C_{t,k}(z)
  =\Pr[T=t,\KA=k,\KB=k\mid z]
  =\Pr[T=t\mid z],
\]
where the last equality follows from \cref{eq:transcript-sum}.  Therefore
\[
  \Pr[\KA=\KB=k\mid T=t,z]=1. \qedhere
\]
\end{proof}

\subsection{Few Candidate Keys}

We next count how many key polynomials can be nonzero for a fixed transcript.

\begin{lemma}[Few candidate keys]\label{cor:few-keys}
For every transcript \(t\), the candidate key set
\[
  \cK_t:=\{k:C_{t,k}\not\equiv0\text{ on }\cube^N\}
\]
has size at most \(2^d\).
\end{lemma}

\begin{proof}
By \Cref{lem:support}, each nonzero \(C_{t,k}\) has support of size at least
\(2^{N-d}\).  By \Cref{lem:disjoint-keys}, these supports are pairwise
disjoint.  Since the cube has \(2^N\) points,
\[
  |\cK_t|\le\frac{2^N}{2^{N-d}}=2^d. \qedhere
\]
\end{proof}

\subsection{A Decision Tree for Disjoint Polynomial Supports}

We now prove the central polynomial statement.  %The proof combines the sign-block, symmetrization, hitting-set, and decision-tree ingredients of Kothari et al.~\cite[Corollary~1, Lemmas~2--3, and Theorem~4]{KothariEtAl26} with the maximum-monomial restriction used by Li et al.~\cite[Lemmas~2.12 and~3.4]{LiEtAl26}. %Its depth depends only on the degree and not on the ambient dimension.

\begin{lemma}[Disjoint-support separator]\label{lem:separator}
Let \(p,q:\cube^N\to\R\) be multilinear polynomials of degree at most \(d\)
such that
\[
  p(z)q(z)=0
  \qquad\text{for every }z\in\cube^N.
\]
There is a deterministic decision tree of depth \(O(d^4)\) which, under the
promise \(p(z)+q(z)\ne0\), determines whether \(p(z)\ne0\) or \(q(z)\ne0\).
\end{lemma}

\begin{proof}
We construct the tree recursively. At one node, restrict \(p\) and \(q\) by
the answers already queried, and let $S\subseteq[N]$ be the set of unqueried
variables. Then we can view $p$ and $q$ as polynomials on the remaining subcube
$\cube^S$. Disjointness and the promise at $\cube^S$ are preserved. If
both restricted polynomials are identically zero, the subcube contains no
promised input and may be labeled arbitrarily. If exactly one is identically
zero, label the node with the other side.

Now suppose both restricted polynomials are not identically zero, and set
\[
  D=\max\{\deg(p),\deg(q)\}.
\]
The case \(D=0\) is impossible, as two nonzero constant functions cannot have
disjoint supports.  Let
\[
  s=p^2-q^2,
\]
so \(\deg(s)\le2D\), and \(s\) is not identically zero on $\cube^S$.
Choose \(z^\star \in\cube^S\) maximizing \(|s(z^\star)|\).  
As $p$ and $q$ has disjoint supports, exactly one of
\(p(z^\star)\) and \(q(z^\star)\) is nonzero.  W.l.o.g., assume that
\(p(z^\star)\ne0\).  Then
\[
  s(z^\star)>0
  \qquad\text{and}\qquad
  q(z^\star)=0.
\]

Let \(z_{M_1}, \ldots, z_{M_b}\) be a maximal family 
of pairwise-disjoint maximum monomials \footnote{Maximal means that for any other maximum monomial \(z_M\) of \(q\), we have $M\cap M_i\ne\emptyset$ for some \(i\in[b]\). Pairwise-disjoint means that \(M_i\cap M_j=\emptyset\) for all \(i\ne j\).}
 of \(q\), and $M_1, \ldots, M_b\subseteq S$ be their corresponding supports. As $q$ is not identically zero, it contains at least one maximum monomial and thus $b\geq 1$.

For each \(i\in[b]\), fix every variable outside
\(M_i\) according to \(z^\star\).  
Such restriction does not change the coefficient of the monomial \(z_{M_i}\),
because $z_{M_i}$ is maximum and no other monomial of \(q\) strictly containing \(M_i\).
Thus the resulting polynomial in the variables
of \(M_i\) is nonzero. Since \(q(z^\star)=0\), there is therefore a nonempty set
\(E_i\subseteq M_i\) such that
\[
  q\bigl((z^\star)^{E_i}\bigr)\ne0.
\]
Disjoint supports of \(p\) and \(q\) imply
\begin{equation} \label{eq:sign-change}
  s(z^\star)>0,
  \qquad
  s\bigl((z^\star)^{E_i}\bigr)<0.
\end{equation}
As $\{M_i\}$ are pairwise disjoint, the blocks \(\{E_i\}\) are also pairwise disjoint.

Now, consider the function $r:\{0,1\}^b\rightarrow \R$ defined by: 
%Then we define a new polynomial \(r\) by restricting \(s\) to simultaneous flips of these blocks and normalize. Namely,
for \(y\in\bits^b\),
\begin{equation*} 
  r(y)=
  \frac{
    s\!\left(
      (z^\star)^{\bigcup_{i:y_i=1}E_i}
    \right)
  }{|s(z^\star)|},
\end{equation*}

Because the blocks are disjoint, every original coordinate is
either fixed or is an affine function of one \(y_i\).  Hence
\(\deg(r)\le \deg(s) \leq 2D\).  The maximality of \(|s(z^\star)|\) and \Cref{eq:sign-change}
gives
\begin{itemize}
\item $|r(y)|\leq 1$ for any $y\in\{0,1\}^b$,
\item $|r(0^b)|=1$,
\item $r(y)\cdot r(0^b)\leq 0$ for any $y\in\{0,1\}^b$ with $|y|=1$.
\end{itemize}
Then, \Cref{lem:discrete-markov} yields $b=O(D^2)$.

Let \(S=M_1\cup\cdots\cup M_b\).  Every \(M_i\) has size
\(\deg(q)\le D\), so
\[
  |S|=O(D^3).
\]
As $\{M_i\}$ are 
maximal, \(S\) intersects the support of every
maximum monomial of \(q\).  Query all variables in \(S\).  Under every
assignment to these variables, each old maximum monomial loses at least
one free variable.  Therefore every branch restriction satisfies
\[
  q|_{\rho}\equiv0
  \qquad\text{or}\qquad
  \deg(q|_{\rho})<\deg(q).
\]
If \(q(z^\star)\ne0\), the symmetric construction instead strictly decreases
the degree of \(p\). 

A branch terminates if either $p$ or $q$ is identically zero after restriction.
Along every nonterminal branch, each stage decreases \(\deg(p)+\deg(q)\) by at
least one. As initially \(\deg(p)+\deg(q)\leq 2d\), and every stage queries
\(O(d^3)\) variables. The resulting decision tree has depth \(O(d^4)\).
\end{proof}

\subsection{The Eavesdropper}

\begin{proof}[Proof of \Cref{thm:main}]
Fix the observed transcript \(t\).  As the eavesdropper is computationally unbounded, it computes the candidate key set
\[
  \cK_t=\{k:C_{t,k}\not\equiv0\}.
\]
This set depends only on the public protocol and \(t\), so requires no queries to the sampled oracle. By \Cref{cor:few-keys}, \(|\cK_t|\le2^d\).  If \(t\) is actually
observed, \Cref{lem:disjoint-keys} guarantees that \(\cK_t\) is nonempty.

The eavesdropper performs balanced binary search on \(\cK_t\).  At a search
node, partition the remaining candidates into two sets \(\cL,\cR\) of as
equal size as possible, and form
\[
  P(z)=\sum_{k\in\cL}C_{t,k}(z),
  \qquad
  Q(z)=\sum_{k\in\cR}C_{t,k}(z).
\]
Both polynomials have degree at most \(d\).  Pairwise disjointness of the key
polynomials gives
\[
  P(z)Q(z)=0
  \qquad\text{for every }z\in\cube^N.
\]
For the actual oracle and every transcript of positive conditional
probability, \Cref{lem:disjoint-keys} guarantees \(P(z)+Q(z)>0\), and exactly one
candidate in \(\cL\cup\cR\) is positive on $z$.  The eavesdropper therefore applies the
decision tree from \Cref{lem:separator} to learn which half contains the candidate.

There are at most
\[
  \left\lceil\log_2|\cK_t|\right\rceil\le d
\]
search levels, and each level uses \(O(d^4)\) oracle-bit queries.  The total
number is \(O(d^5)\).  At the final leaf, the eavesdropper outputs the sole
remaining key.  By \Cref{lem:disjoint-keys}, this is exactly the key output by
both honest parties, for every \(z,t\) of positive conditional probability.
Therefore
\[
  \Pr[\KA=\KB=\KE]=1.
\]
Finally, \(d\le2(q_{\A}+q_{\B})\), so the total query complexity is
\(O((q_{\A}+q_{\B})^5)\).  Every query reads one oracle bit \(H(x)\), and is
therefore a classical query permitted in the quantum-query model.
\end{proof}

\bibliographystyle{alpha}
\bibliography{references}

\end{document}